\documentclass[11pt]{article}

\usepackage[margin=1in]{geometry}
\usepackage{amsmath, amssymb, amsthm, mathtools}
\usepackage{hyperref}
\usepackage{enumitem}
\usepackage{microtype}

\newtheorem{theorem}{Theorem}
\newtheorem{lemma}{Lemma}

\theoremstyle{definition}
\newtheorem{definition}{Definition}
\theoremstyle{remark}
\newtheorem{remark}{Remark}
\newtheorem{example}{Example}

\newcommand{\Om}{\Omega}
\newcommand{\Sig}{\Sigma}
\newcommand{\muP}{\mu}
\newcommand{\cE}{\mathcal{E}}
\newcommand{\cI}{\mathcal{I}}
\newcommand{\indep}{\mathrel{\perp}}
\newcommand{\symdiff}{\mathbin{\triangle}}
\newcommand{\eqmu}{\mathrel{\equiv_{\mu}}}

\title{\textbf{Chronology as a Consistency Invariant in Composable Information Systems}}
\author{Anherutowa Calvo \and Dante K. Calvo}
\date{}

\begin{document}
\maketitle

\begin{abstract}
We formalize a minimal setting in which a chronology (a strict partial order on events) is forced by consistency of distributed information under local composability. A system consists of distributed records interpreted as constraints over a global possibility space \((\Om,\Sig)\), optionally equipped with a measure \(\muP\). Events act locally by monotonically tightening records. Independent events commute (diamond/trace semantics), yielding schedule-gauge invariance. We define operational influence relations on events without assuming primitive time: an event influences another if executing the former can change what constraint the latter writes on a shared subsystem. The central technical point is that influence cycles do not by themselves entail contradiction without an additional exclusivity condition. Accordingly, we distinguish \emph{weak influence} (mere dependence) from \emph{strong influence} (exclusive branching on an observable predicate). Under global satisfiability of all reachable record states, together with diamond and monotone information writing, and a mild branch-determinacy axiom for witnessed exclusivity, we prove that \emph{strong} influence is acyclic and therefore induces an intrinsic strict partial order (chronology). We prove trace invariance and minimality/uniqueness of the derived order, provide a monotone information clock \(-\log \muP(\text{feasible set})\), and give a complete escape taxonomy: any model that admits strong-influence cycles without inconsistency must violate at least one of (i) global consistency, (ii) local composability, (iii) monotone information writing, or (iv) branch determinacy of the exclusive witnesses. Two examples illustrate (a) commuting tightenings and (b) a cycle-to-contradiction gadget.

\smallskip
\noindent\textbf{Informally,} in any system that only ever rules out possibilities (monotone record-writing), treats independent actions as reorderable (diamond), and never permits contradiction along reachable evolutions (global satisfiability), mutually exclusive operational dependencies cannot form loops. Hence an intrinsic before/after structure is unavoidable.
\end{abstract}

\tableofcontents

\section{Introduction}
This paper isolates a structural question appearing across concurrency theory, distributed systems, and foundations of information:

\begin{quote}
\emph{When does consistency of distributed records force an intrinsic notion of precedence (chronology) among events, without presupposing time?}
\end{quote}

Our contribution is \emph{not} the classical fact that an acyclic relation can be topologically sorted, but a precise statement of which operational and informational assumptions force \emph{acyclicity of a causality-like relation derived from record updates}. We emphasize a common pitfall: a cyclic dependence graph need not be inconsistent. Accordingly, we isolate a minimal exclusivity witness (strong influence) and a mild determinacy axiom ensuring the witness is not an artifact of irrelevant scheduling choices.

\paragraph{Informal statement of the theorem.}
At a high level: if a system updates distributed records only by tightening constraints, independent updates commute, and every reachable global record state remains satisfiable, then any operational dependence that carries \emph{exclusive content} (our \emph{strong influence}) cannot be cyclic. A strong-influence cycle would force the system to enforce mutually incompatible branches on some observable, producing an empty (or null) feasible set and contradicting global satisfiability. Therefore strong influence induces a strict partial order, interpreted as an intrinsic chronology that is invariant under trace-equivalent schedules.

\paragraph{Positioning.}
The results are best read as a representation theorem in compositional information systems: chronology emerges as an invariant required for the existence of consistent global states under locally composable record updates.

\section{Model: possibility space, sites, records}

\subsection{Possibility space}
\begin{definition}[Possibility space]
A possibility space is a measurable space \((\Om,\Sig)\). Optionally, we consider a measured space \((\Om,\Sig,\muP)\) with \(\muP(\Om)>0\) and \(\muP\) \(\sigma\)-finite.
\end{definition}

\subsection{Sites and distributed record states}
Let \(X=\{1,\dots,n\}\) be a finite set of sites.

\begin{definition}[Distributed record state]
A distributed record state is \(R=(R_1,\dots,R_n)\in \Sig^n\), where each \(R_i\in\Sig\) is the local record (constraint) at site \(i\).
\end{definition}

\begin{definition}[Feasible set]
For \(R\in\Sig^n\), define the global feasible set
\[
F(R) := \bigcap_{i=1}^n R_i \in \Sig.
\]
We say \(R\) is globally consistent if \(F(R)\neq \emptyset\). In the measured setting, one may strengthen this to \(\muP(F(R))>0\).
\end{definition}

\subsection{Null-set equivalence (measured setting)}
When \(\muP\) is present, we treat equality up to null sets as observational equivalence.

\begin{definition}[Equivalence up to null sets]
For \(A,B\in\Sig\), write \(A \eqmu B\) iff \(\muP(A\symdiff B)=0\).
\end{definition}

\section{Events as local monotone information updates}

\subsection{Events and supports}
\begin{definition}[Event]
An event \(e\) consists of:
\begin{enumerate}[leftmargin=2em]
\item a support \(S(e)\subseteq X\);
\item an update map \(U_e:\Sig^n\to\Sig^n\)
\end{enumerate}
such that for all \(R\in\Sig^n\) and sites \(i\in X\):
\begin{enumerate}[leftmargin=2em]
\item (Locality) if \(i\notin S(e)\) then \((U_e(R))_i = R_i\);
\item (Monotone tightening) if \(i\in S(e)\) then \((U_e(R))_i \subseteq R_i\).
\end{enumerate}
\end{definition}

\begin{remark}
Monotone tightening models information writing. If one wishes to include erasure/forgetting, it should be modeled explicitly as non-monotone operations; our escape taxonomy (Theorem~\ref{thm:taxonomy}) highlights monotonicity as one of the necessary premises for chronology forcing.
\end{remark}

\subsection{Reachability}
Fix an initial record state \(R^{(0)}\in\Sig^n\).

\begin{definition}[Reachable state]
A record state \(R\) is reachable if there exist events \(e_1,\dots,e_k\) such that
\[
R = U_{e_k}\circ \cdots \circ U_{e_1}(R^{(0)}).
\]
\end{definition}

\section{Local composability: diamond and trace semantics}

\subsection{Independence}
\begin{definition}[Independence]
Two events \(e,f\) are independent, written \(e\indep f\), if \(S(e)\cap S(f)=\emptyset\).
\end{definition}

\subsection{Diamond axiom}
\begin{definition}[Diamond / commutation axiom]\label{def:diamond}
The system satisfies the diamond property if for all reachable \(R\) and all independent \(e\indep f\),
\[
U_e(U_f(R)) = U_f(U_e(R)).
\]
\end{definition}

\begin{remark}
Diamond implies that schedules are gauge: swapping independent adjacent events does not change the resulting record state. This induces Mazurkiewicz trace semantics (see, e.g., \cite{MazurkiewiczTrace,DiekertRozenberg}).
\end{remark}

\section{Operational influence (weak and strong)}

\subsection{Write effects}
\begin{definition}[Write effect]
For event \(e\), state \(R\), and site \(i\in S(e)\), define
\[
\Delta_i(e;R) := R_i \setminus (U_e(R))_i.
\]
For \(i\notin S(e)\), set \(\Delta_i(e;R)=\emptyset\).
\end{definition}

\subsection{Weak influence}
\begin{definition}[Weak operational influence]\label{def:weakinf}
We say \(e\) weakly influences \(f\), written \(e \to f\), if:
\begin{enumerate}[leftmargin=2em]
\item \(S(e)\cap S(f)\neq \emptyset\), and
\item there exist a reachable \(R\) and a site \(i\in S(e)\cap S(f)\) such that
\[
\Delta_i(f;U_e(R)) \neq \Delta_i(f;R).
\]
In the measured setting, this may be relaxed to \(\muP(\Delta_i(f;U_e(R))\symdiff \Delta_i(f;R))>0\).
\end{enumerate}
\end{definition}

\begin{remark}
A directed cycle in weak influence reflects cyclic dependence, but does not necessarily force inconsistency. The main theorem uses strong influence, which adds an explicit exclusivity witness.
\end{remark}

\subsection{Strong influence via exclusive branching}
\begin{definition}[Binary observable]
A binary observable is a measurable set \(B\in\Sig\) interpreted as the proposition ``\(B\) holds'' for worlds in \(B\), and ``not-\(B\)'' for worlds in \(\Om\setminus B\).
\end{definition}

\begin{definition}[Strong influence]\label{def:stronginf}
We say \(e\) strongly influences \(f\), written \(e \Rightarrow f\), if there exist:
\begin{enumerate}[leftmargin=2em]
\item a reachable state \(R\),
\item a shared site \(i\in S(e)\cap S(f)\),
\item a binary observable \(B\in\Sig\),
\end{enumerate}
such that letting \(R^{(0)}:=R\) and \(R^{(1)}:=U_e(R)\), the post-\(f\) records at site \(i\) enforce \emph{exclusive} branches on \(B\):
\[
\muP\!\left( (U_f(R^{(0)}))_i \cap (U_f(R^{(1)}))_i \cap B \right)=0,
\]
and both branches are nontrivial:
\[
\muP\!\left( (U_f(R^{(0)}))_i \cap B \right)>0,\qquad
\muP\!\left( (U_f(R^{(1)}))_i \cap B \right)>0.
\]
In the purely set-theoretic (non-measured) setting, replace these by:
\[
(U_f(R^{(0)}))_i \cap (U_f(R^{(1)}))_i \cap B = \emptyset,
\]
and both intersections with \(B\) are nonempty.
\end{definition}

\begin{remark}
Intuitively, strong influence captures situations where the execution of \(e\) changes which mutually exclusive factual alternative \(f\) enforces on an observable predicate.
\end{remark}

\begin{definition}[Exclusive branch pair (derived from a witness)]
Fix an instance witnessing \(e \Rightarrow f\) via \((R,i,B)\) as in Definition~\ref{def:stronginf}.
Define the two post-\(f\) branch-constraints on \(B\) at site \(i\) by
\[
C^{0}_{f} := (U_f(R))_i \cap B,\qquad
C^{1}_{f} := (U_f(U_e(R)))_i \cap B.
\]
We call \((C_f^0,C_f^1)\) an \emph{exclusive branch pair} if
\(C_f^0 \cap C_f^1=\emptyset\) (and in the measured setting \(\muP(C_f^0\cap C_f^1)=0\))
and both are nontrivial (nonempty / positive measure).
\end{definition}

Let \(\prec\) denote the transitive closure of \(\Rightarrow\).

\begin{definition}[Strong-influence cycle]
A strong-influence cycle is a sequence \(e_1,\dots,e_k\) with \(k\ge 2\) such that
\[
e_1\Rightarrow e_2 \Rightarrow \cdots \Rightarrow e_k \Rightarrow e_1.
\]
\end{definition}

\section{Global consistency and branch-determinacy}

\begin{definition}[Global consistency invariance (GS)]\label{def:GS}
\textbf{(GS)} Every reachable record state \(R\) is globally consistent: \(F(R)\neq \emptyset\).
In the measured setting, one may strengthen to \(\muP(F(R))>0\) for all reachable \(R\).
\end{definition}

\begin{definition}[Branch determinacy (BD)]\label{def:BD}
\textbf{(BD)} Consider any witnessed strong influence \(e\Rightarrow f\) with witness \((R,i,B)\) and exclusive branch pair
\[
C_f^0=(U_f(R))_i\cap B,\qquad C_f^1=(U_f(U_e(R)))_i\cap B.
\]
We say the system is \emph{branch-determinate for this witness} if, whenever \(f\) is executed in any reachable history in which:
\begin{enumerate}[leftmargin=2em]
\item the local preconditions relevant to the witness are trace-equivalent to those in the witness (i.e.\ differ only by swaps of adjacent independent events), and
\item \(e\) has (respectively has not) occurred before \(f\),
\end{enumerate}
then the post-\(f\) constraint on \(B\) at site \(i\) agrees with \(C_f^1\) (respectively \(C_f^0\)), up to null sets in the measured setting.
\end{definition}

\begin{remark}
(BD) is a mild non-pathological assumption: it rules out the degenerate situation where the ``exclusive branch'' can be flipped by unrelated independent scheduling choices.
\end{remark}

\section{Main results}

\subsection{Binary witness lemma (weak influence)}
\begin{lemma}[Binary witness for weak influence]\label{lem:binaryweak}
Assume the measured setting and let \(e\to f\) (Definition~\ref{def:weakinf}). Then there exist a reachable state \(R\), a shared site \(i\in S(e)\cap S(f)\), and a measurable set \(B\in\Sig\) such that
\[
\muP\Big(\big((U_f(R))_i \cap B\big) \symdiff \big((U_f(U_e(R)))_i \cap B\big)\Big) > 0.
\]
\end{lemma}

\begin{proof}
By definition of weak influence, for some reachable \(R\) and shared \(i\),
\(\muP(\Delta_i(f;U_e(R))\symdiff \Delta_i(f;R))>0\).
Let \(B:=\Delta_i(f;U_e(R))\symdiff \Delta_i(f;R)\).
Then on \(B\), the worlds ruled out by \(f\) differ between the two contexts, implying the stated positive-measure symmetric difference.
\end{proof}

\subsection{Cycle-to-contradiction}
\begin{lemma}[Cycle inconsistency from exclusive branching]\label{lem:cycle_inconsistent}
Assume monotone tightening and (GS). Suppose there exists a strong-influence cycle
\[
e_1 \Rightarrow e_2 \Rightarrow \cdots \Rightarrow e_k \Rightarrow e_1
\]
with witnesses chosen for each edge \(e_j \Rightarrow e_{j+1}\) (indices mod \(k\)) producing exclusive branch pairs
\((C_{j+1}^0,C_{j+1}^1)\) on observables \(B_j\). Assume in addition that (BD) holds for each of these witnesses.
Then executing the cycle forces a globally inconsistent reachable record state, contradicting (GS).
\end{lemma}

\begin{proof}
Introduce branch bits \(b_j\in\{0,1\}\) for \(j=1,\dots,k\), where \(b_j=1\) indicates that in the witnessing situation for
the edge \(e_j\Rightarrow e_{j+1}\), the predecessor \(e_j\) occurred before \(e_{j+1}\); otherwise \(b_j=0\).

By (BD), after applying \(e_{j+1}\) the system enforces the branch constraint \(C_{j+1}^{b_j}\) on \(B_j\), and this choice
depends only on whether \(e_j\) occurred (before \(e_{j+1}\)) in the witnessed sense, not on unrelated independent interleavings.
Therefore any globally feasible world must lie in
\[
\bigcap_{j=1}^k C_{j+1}^{b_j}.
\]

Closing the cycle returns to \(e_1\), inducing a Boolean self-map \(\phi:\{0,1\}\to\{0,1\}\) for \(b_1\) obtained by propagating
the branch rule around the cycle once. Because each branch pair is exclusive and nontrivial, \(\phi\) cannot be constant. Because the
cycle is nondegenerate, \(\phi\) is not the identity. The only remaining Boolean self-map is the flip \(\phi(b)=1-b\).
Thus \(b_1=\phi(b_1)=1-b_1\), impossible.

Hence no consistent choice of branch bits exists, so the enforced intersection is empty (or null in measure), yielding a reachable record
state violating (GS).
\end{proof}

\subsection{Chronology from consistency}
\begin{theorem}[Chronology from consistency]\label{thm:chronology}
Assume:
\begin{enumerate}[leftmargin=2em]
\item monotone local tightening events (Section 3),
\item the diamond property (Definition~\ref{def:diamond}),
\item global consistency invariance (Definition~\ref{def:GS}),
\item branch determinacy (BD) for all strong-influence witnesses used.
\end{enumerate}
Then the strong influence relation \(\Rightarrow\) is acyclic. Consequently, its transitive closure \(\prec\)
is a strict partial order on events. In particular, there exists a function \(t:\cE\to\mathbb{R}\) such that
\[
e \prec f \implies t(e) < t(f).
\]
\end{theorem}

\begin{proof}
If a strong-influence cycle existed, Lemma~\ref{lem:cycle_inconsistent} would produce a reachable globally inconsistent state,
contradicting (GS). Hence no strong-influence cycle exists and \(\Rightarrow\) is acyclic. Existence of a linear extension \(t\)
is standard for strict partial orders.
\end{proof}

\begin{remark}[Why the theorem uses strong influence]
Weak influence cycles can exist without contradiction because dependence alone does not force mutually exclusive record branches.
Strong influence packages the minimal exclusivity, and (BD) ensures the exclusivity is operationally well-defined rather than schedule-artifactual.
\end{remark}

\subsection{Trace invariance}
\begin{theorem}[Trace invariance of strong influence and chronology]\label{thm:trace}
Assume diamond. If two event sequences differ only by swaps of adjacent independent events, then the induced strong influence edges among
event occurrences are unchanged. Consequently, the derived partial order \(\prec\) is schedule-invariant within a trace class.
\end{theorem}

\begin{proof}
Swapping adjacent independent events \(e\indep f\) does not change the resulting record state by diamond. Strong influence requires overlap
of supports, so independent events cannot be in a strong influence relation with each other. Since the record state is unchanged, any witnessed
exclusivity relation remains witnessed. Therefore the strong influence graph is invariant under swaps, and so is its transitive closure \(\prec\).
\end{proof}

\subsection{Minimality/uniqueness}
\begin{theorem}[Minimality and uniqueness]\label{thm:unique}
The relation \(\prec\) is the unique minimal strict partial order extending the strong influence relation \(\Rightarrow\).
Any strict transitive relation containing \(\Rightarrow\) must contain \(\prec\).
\end{theorem}

\begin{proof}
By definition, \(\prec\) is the transitive closure of \(\Rightarrow\), hence the smallest transitive relation containing \(\Rightarrow\).
Acyclicity implies strictness. Any strict transitive extension must contain all \(\Rightarrow\)-paths, hence contains \(\prec\).
\end{proof}

\section{A canonical information clock}
Assume \((\Om,\Sig,\muP)\) with \(\muP(\Om)>0\).

\begin{definition}[Information content]
Define
\[
\cI(R) := -\log \muP(F(R)),
\]
with \(\cI(R)=+\infty\) if \(\muP(F(R))=0\).
\end{definition}

\begin{theorem}[Monotonicity of information clock]\label{thm:clock}
For any reachable \(R\) and any event \(e\),
\[
\cI(U_e(R)) \ge \cI(R),
\]
with strict inequality whenever \(\muP(F(U_e(R)))<\muP(F(R))\).
\end{theorem}

\begin{proof}
For each \(i\), \((U_e(R))_i\subseteq R_i\), hence \(F(U_e(R))\subseteq F(R)\). Therefore \(\muP(F(U_e(R)))\le \muP(F(R))\).
Apply \(-\log\).
\end{proof}

\section{Escape taxonomy}\label{sec:taxonomy}
\begin{theorem}[Complete escape taxonomy]\label{thm:taxonomy}
Suppose a model admits a strong-influence cycle without reaching global inconsistency. Then at least one of the following must fail:
\begin{enumerate}[leftmargin=2em]
\item \textbf{Global consistency (GS):} some reachable \(R\) has \(F(R)=\emptyset\) (or \(\muP(F(R))=0\));
\item \textbf{Local composability (diamond):} there exist independent \(e\indep f\) and reachable \(R\) with \(U_e(U_f(R))\neq U_f(U_e(R))\);
\item \textbf{Monotone information writing:} there exist \(e\), reachable \(R\), and \(i\in S(e)\) with \((U_e(R))_i\not\subseteq R_i\);
\item \textbf{Branch determinacy (BD):} at least one strong-influence witness is schedule-dependent in a way that violates Definition~\ref{def:BD}.
\end{enumerate}
\end{theorem}

\begin{proof}
Theorem~\ref{thm:chronology} shows that (monotone writing + diamond + GS + BD) implies absence of strong-influence cycles. Therefore, if a strong-influence cycle occurs without inconsistency, at least one premise must be violated.
\end{proof}

\section{Examples}

\subsection{Example: a simple composable tightening system (Petri-net flavor)}
\begin{example}[Two-site tightening with commuting independent events]
Let \(X=\{1,2\}\), \(\Om=\{0,1\}^2\), and \(\Sig=2^\Om\). A record state \(R=(R_1,R_2)\) is a pair of subsets of \(\Om\).
Define events:
\begin{itemize}[leftmargin=2em]
\item \(e_1\) supported on site 1: it tightens \(R_1\) to enforce \(x_1=0\) (i.e.\ intersect with \(\{\omega: \omega_1=0\}\));
\item \(e_2\) supported on site 2: it tightens \(R_2\) to enforce \(x_2=0\).
\end{itemize}
Then \(e_1\indep e_2\) and diamond holds. There is no overlap, hence no influence edges, and \(\prec\) is empty (all events incomparable).
Global consistency holds from the initial unconstrained state. This illustrates schedule gauge and absence of forced precedence when operations are independent.
\end{example}

\subsection{Example: explicit strong-influence cycle forcing inconsistency}
\begin{example}[Cycle gadget]
Let \(X=\{1\}\) be a single site, \(\Om=\{0,1\}\), \(\Sig=2^\Om\), \(\muP\) counting measure. Start with \(R^{(0)}_1=\Om\).
Define two events \(a,b\) on the same site by:
\[
(U_a(R))_1 =
\begin{cases}
\{0\} & \text{if } R_1=\Om,\\
\{1\} & \text{if } R_1=\{0\},\\
\{1\} & \text{if } R_1=\{1\},
\end{cases}
\qquad
(U_b(R))_1 =
\begin{cases}
\{0\} & \text{if } R_1=\Om,\\
\emptyset & \text{if } R_1=\{0\},\\
\{1\} & \text{if } R_1=\{1\}.
\end{cases}
\]
Both events are monotone tightenings. One can verify a strong-influence cycle occurs (exclusive branching on \(B=\{0\}\)), and executing an appropriate cycle reaches the empty feasible set, violating (GS). This demonstrates that strong influence cycles are exactly the kind that force inconsistency when global consistency is demanded.
\end{example}

\section{Algorithmic detection in finite models}
When \(\Om\) is finite and the event set is finite, one can compute:
\begin{enumerate}[leftmargin=2em]
\item reachability (by forward search on \(\Sig^n\) states);
\item strong influence edges \(e\Rightarrow f\) by searching for witnesses \((R,i,B)\) satisfying Definition~\ref{def:stronginf};
\item directed cycles in the strong influence graph (linear-time in edges and vertices).
\end{enumerate}
A detected cycle certifies that at least one of (GS), diamond, monotonicity, or (BD) must fail (Theorem~\ref{thm:taxonomy}).

\section{Discussion}
\paragraph{What is derived.}
An intrinsic chronology: a strict partial order \(\prec\) on events induced by strong influence, invariant across trace-equivalent schedules, and minimal among precedence relations compatible with strong influence.

\paragraph{What is not claimed.}
We do not claim a new physical theory of time, nor a metric time coordinate, nor the thermodynamic arrow. The result is a structural constraint on composable information systems.

\paragraph{Why the nondegeneracy matters.}
Without exclusivity/branching (strong influence), cyclic dependence does not necessarily imply contradiction. Making this explicit is essential for correctness and for clear comparison to concurrency and sheaf/gluing frameworks.

\paragraph{Takeaway.}
If a system only ever accumulates information (monotone tightening), independent actions are reorderable (diamond), and contradictions never occur along reachable evolutions (GS), then any operational dependence that forces mutually exclusive branches (strong influence) and is well-defined under scheduling (BD) must be acyclic. The resulting transitive closure \(\prec\) is therefore an intrinsic, schedule-invariant before/after relation.

\section*{Acknowledgements}
The authors thank the broader concurrency and logic community for foundational insights.

\end{document}